\documentclass[11pt]{article}
\usepackage{graphicx,amsmath,amsbsy,amssymb,amsthm,citesort,epsfig,epsf,url,enumerate,comment}
\usepackage[usenames]{color}
\usepackage{bibspacing}
\def\real    { \mathbb{R} }

\newtheorem{thm}{Theorem} 
\newtheorem{lemma}{Lemma} 
\newtheorem{cor}{Corollary}

\newcommand{\bitem}{\begin{itemize}}
\newcommand{\eitem}{\end{itemize}}

\newcommand{\supp}{\mathrm{supp}}

\newcommand{\beqn}{\begin{equation}}
\newcommand{\eeqn}{\end{equation}}
\newcommand{\balign}{\begin{align}}
\newcommand{\ealign}{\end{align}}

\newcommand{\cN}{\mathcal{N}}
\newcommand{\cP}{\mathcal{P}}

\newcommand{\by}{\boldsymbol{y}}
\newcommand{\bx}{\boldsymbol{x}}

\newcommand{\bX}{\boldsymbol{X}}
\newcommand{\bA}{\boldsymbol{A}}
\newcommand{\bz}{\boldsymbol{z}}

\newcommand{\bw}{\boldsymbol{w}}

\newcommand{\bU}{\boldsymbol{U}}
\newcommand{\bSigma}{\boldsymbol{\Sigma}}
\newcommand{\bV}{\boldsymbol{V}}
\newcommand{\bI}{\boldsymbol{I}}
\newcommand{\bmu}{\boldsymbol{\mu}}

\newcommand{\bQ}{\boldsymbol{Q}}

\newcommand{\bzero}{\boldsymbol{0}}

\newcommand{\xhat}{\widehat{\bx}}

\newcommand{\expval}[2][]{\mathbb{E}_{#1} \left[ #2 \right]}
\newcommand{\prob}[2][]{\mathbb{P}_{#1} \left[ #2 \right]}
\newcommand{\norm}[2][2]{\left\| #2 \right\|_{#1}}

\newcommand{\absval}[1]{\left| #1 \right|}

\newcommand{\tr}[1]{\mathrm{Tr}\left( #1 \right)}

\newcommand{\Pset}{\mathcal{X}}
\newcommand{\restrict}[1]{\vert_{#1}}

\begin{document}

\title{How well can we estimate a sparse vector?}

\author{Emmanuel J. Cand\`{e}s and Mark A. Davenport\thanks{E.J.C.\
    is with the Departments of Mathematics and Statistics, Stanford
    University, Stanford, CA 94035. M.A.D.\ is with the School of Electrical and Computer Engineering, Georgia Institute of Technology, Atlanta, GA 30332. This work has
    been partially supported by NSF grant DMS-1004718, the Waterman
    Award from NSF, ONR grant N00014-10-1-0599 and a grant from AFOSR.
    M.A.D.\ is the corresponding author.  Email: mdav@gatech.edu.  Phone: (404)894-2881.  Fax: (404)894-8363.}}

\maketitle

\begin{abstract}
The estimation of a sparse vector in the linear model is a fundamental problem in signal processing, statistics, and compressive sensing.  This paper establishes a lower bound on the mean-squared error, which holds regardless of the sensing/design matrix being used and regardless of the estimation procedure. This lower bound very nearly matches the known upper bound one gets by taking a random projection of the sparse vector followed by an $\ell_1$ estimation procedure such as the Dantzig selector. In this sense, compressive sensing techniques cannot essentially be improved.
\end{abstract}

\begin{center}
lower bounds, Fano's inequality, matrix Bernstein inequality
\end{center}

\section{Introduction}

The estimation of a sparse vector from noisy observations is a
fundamental problem in signal processing and statistics, and lies at
the heart of the growing field of compressive
sensing~\cite{Donoh_Compressed,CandeRT_Robust,CandeT_Near}. At its
most basic level, we are interested in accurately estimating a vector
$\bx \in \real^n$ that has at most $k$ non-zeros from a set of noisy
linear measurements
\begin{equation} \label{eq:measmodel1}
\by = \bA \bx + \bz,
\end{equation}
where $\bA \in \real^{m \times n}$ and $\bz \sim \cN(\bzero,\sigma^2
\bI)$.  We are often interested in the underdetermined setting where
$m$ may be much smaller than $n$.  In general, one would not expect to
be able to accurately recover $\bx$ when $m < n$ since there are more
unknowns than observations.  However it is by now well-known that by
exploiting sparsity, it is possible to accurately estimate $\bx$.

As an example, consider what is known concerning $\ell_1$ minimization techniques, which are among the most powerful and well-understood with respect to their performance in noise.  Specifically, if we suppose that the entries of the matrix $\bA$ are i.i.d.\ $\cN(0,1/n)$, then one can show that for any $\bx \in \Sigma_k :=  \{ \bx : \norm[0]{\bx} \le k \}$, $\ell_1$ minimization techniques such as the Lasso or the Dantzig selector produce a recovery $\xhat$ such that
\begin{equation} \label{eq:ds_error}
\frac{1}{n} \norm{\xhat - \bx }^2 \le C_0 \frac{k \sigma^2}{m} \log n
\end{equation}
holds with high probability provided that $m = \Omega \left( k \log(n/k) \right)$~\cite{CandeT_Dantzig}. We refer to
\cite{BickeRT_Simultaneous} and \cite{DonohJMM_Compressed} for further results.

\subsection{Criticism}

A noteworthy aspect of the bound in (\ref{eq:ds_error}) is that the recovery error increases linearly as we decrease $m$, and thus we pay a penalty for taking a small number of measurements.  Although this effect is sometimes cited as a drawback of the compressive sensing framework, it should not be surprising --- we fully expect that if each measurement has a constant SNR, then taking more measurements should reduce our estimation error.

However, there is another somewhat more troubling aspect of (\ref{eq:ds_error}).  Specifically, by filling the rows of $\bA$ with i.i.d.\ random variables, we are ensuring that our ``sensing vectors'' are almost orthogonal to our signal of interest, leading to a tremendous SNR loss.  To quantify this loss, suppose that we had access to an oracle that knows {\em a priori} the locations of the nonzero entries of $\bx$ and could instead construct $\bA$ with vectors localized to the support of $\bx$.  For example, if $m$ is an integer multiple of $k$ then we could simply measure sample each coefficient directly $m/k$ times and then average these samples.  One can check that this procedure would yield an estimate obeying
\begin{equation} \label{eq:oracle_error}
\expval{\frac{1}{n} \norm{\xhat - \bx }^2} =  \left(\frac{k \sigma^2}{m}\right)  \left(\frac{k}{n}\right).
\end{equation}
Thus, the performance in (\ref{eq:ds_error}) is worse than what would
be possible with an oracle by a factor of $(n/k)\log n$. When $k$ is
small, this is very large!  Of course, we won't have access to an
oracle in practice, but the substantial difference between
(\ref{eq:ds_error}) and (\ref{eq:oracle_error}) naturally leads one to
question whether (\ref{eq:ds_error}) can be improved upon.

\subsection{Can we do better?}

In this paper we will approach this question from the viewpoint of
compressive sensing and/or of experimental design.  Specifically, we
assume that we are free to choose {\em both} the matrix $\bA$ {\em
  and} the sparse recovery algorithm.  Our results will have
implications for the case where $\bA$ is determined by factors beyond
our control, but our primary interest will be in considering the
performance obtained by the best possible choice of $\bA$.  In this
setting, our fundamental question is:
\begin{quote}
{\em
Can we ever hope to do better than (\ref{eq:ds_error})?  Is there a more intelligent choice for the matrix $\bA$?  Is there a more effective recovery algorithm?}
\end{quote}
In this paper we show that the answer is {\em no}, and that there exists no choice of $\bA$ or recovery algorithm that can significantly improve upon the guarantee in (\ref{eq:ds_error}). Specifically, we consider the worst-case error over all $\bx \in \Sigma_k$, i.e.,
\begin{equation} \label{eq:minimaxrisk}
M^*(\bA) = \inf_{\xhat} \sup_{\bx \in \Sigma_k} \expval{ \frac{1}{n} \norm{\xhat(\by) - \bx}^2}.
\end{equation}
Our main result consists of the following bound, which establishes a fundamental limit on the minimax risk which holds for any matrix $\bA$ and any possible recovery algorithm.
\begin{thm} \label{thm1}
Suppose that we observe $\by = \bA \bx + \bz$ where $\bx$ is a $k$-sparse vector, $\bA$ is an $m \times n$ matrix with $m \ge k$, and $\bz \sim \cN(\bzero,\sigma^2 \bI)$. Then there exists a constant $C_1>0$ such that for all $\bA$,
\begin{equation} \label{eq:mmriskbound_a}
M^*(\bA) \ge C_1 \frac{k \sigma^2}{\norm[F]{\bA}^2} \log \left(n/k\right).
\end{equation}
We also have that for all $\bA$
\begin{equation} \label{eq:mmriskbound_b}
M^*(\bA) \ge \frac{k \sigma^2}{\norm[F]{\bA}^2}.
\end{equation}
\end{thm}

This theorem says that {\em there is no $\bA$} and {\em no recovery
  algorithm} that does fundamentally better than the Dantzig selector
(\ref{eq:ds_error}) up to a constant\footnote{Our analysis shows that
  asymptotically $C_1$ can be taken as $1/128$.  We have made no
  effort to optimize this constant, and it is probably far from sharp.
  This is why we give the simpler bound (\ref{eq:mmriskbound_b}) which
  is proven by considering the error we would incur even if we knew
  the support of $\bx$ {\em a priori}.  However, our main result is
  (\ref{eq:mmriskbound_a}).  We leave the calculation of an improved
  constant to future work.}; that is, ignoring the difference in the
factors $\log n/k$ and $\log n$.  In this sense, the results of
compressive sensing are at the limit.

Although the noise model in (\ref{eq:measmodel1}) is fairly common, in some settings (such as the estimation of a signal transmitted over a noisy channel) it is more natural to consider noise that has been added directly to the signal prior to the acquisition of the measurements.  In this case we can directly apply Theorem~\ref{thm1} to obtain the following corollary.

\begin{cor} \label{cor1}
Suppose that we observe $\by = \bA(\bx + \bw)$ where $\bx$ is a $k$-sparse vector, $\bA$ is an $m \times n$ matrix with $k \le m \le n$, and $\bw  \sim \cN(\bzero,\sigma^2 \bI)$.  Then for all $\bA$
\begin{equation} \label{eq:mmriskbound2}
M^*(\bA) \ge C_1 \frac{k \sigma^2}{m} \log \left(n/k\right) \quad \quad \mathrm{and} \quad \quad M^*(\bA) \ge \frac{k \sigma^2}{m}.
\end{equation}
\end{cor}
\begin{proof}
We assume that $\bA$ has rank $m' \le m$.  Let $\bU \bSigma \bV^*$ be the reduced SVD of $\bA$, where $\bU$ is $m \times m'$, $\bSigma$ is $m' \times m'$, and $\bV$ is $n \times m'$. Applying the matrix $ \bSigma^{-1} \bU^*$ to $\by$ preserves all the information about $\bx$, and so we can equivalently assume that the data is given by
\begin{equation}
\by' = \bSigma^{-1} \bU^* \by = \bV^* \bx + \bV^* \bw.
\end{equation}
Note that $\bV^* \bw$ is a Gaussian vector with covariance matrix $\sigma^2 \bV^* \bV = \sigma^2 \bI$.  Moreover, $\bV^*$ has unit-norm rows, so that $\norm[F]{\bV^*} \le m' \le m$.  We then apply Theorem~\ref{thm1} to establish (\ref{eq:mmriskbound2}).
\end{proof}

The intuition behind this result is that when noise is added to the measurements, we can boost the SNR by rescaling $\bA$ to have higher norm.  When we instead add noise to the signal, the noise is also scaled by $\bA$, and so no matter how $\bA$ is designed there will always be a penalty of $1/m$.

\subsection{Related work}

There have been a number of prior works that have established lower bounds on $M^*(\bA)$ or related quantities under varying assumptions~\cite{YeZ_Rate,RaskuWY_Minimax,RaskuWY_Lower,Louni_Generalized,RigolT_Exponential,AeroSZ_Information,SarvoBB_Measurements}.  In~\cite{AeroSZ_Information,SarvoBB_Measurements}, techniques from information theory similar to the ones that we use below are used to establish rather general lower bounds under the assumption that the entries of $\bx$ are generated i.i.d.\ according to some distribution.  For an appropriate choice of distribution, $\bx$ will be approximately sparse and~\cite{AeroSZ_Information,SarvoBB_Measurements} will yield asymptotic lower bounds of a similar flavor to ours.

The prior work most closely related to our results is that of Ye and Zhang~\cite{YeZ_Rate} and Raskutti, Wainwright, and Yu~\cite{RaskuWY_Minimax}.  In~\cite{YeZ_Rate} Ye and Zhang establish a bound similar to~(\ref{eq:mmriskbound_a}) in Theorem 1.  While the resulting bounds are substantially the same, the bounds in~\cite{YeZ_Rate} hold only in the asymptotic regime where $k \rightarrow \infty$, $n \rightarrow \infty$, and $\frac{k}{n} \rightarrow 0$, whereas our bounds hold for arbitrary finite values of $k$ and $n$, including the case where $k$ is relatively large compared to $n$.  In~\cite{RaskuWY_Minimax} Raskutti et al.\ reach a somewhat similar conclusion to our Theorem~\ref{thm1} via a similar argument, but where it is assumed that $\bA$ satisfies $\norm{\bA \bx}^2 \le (1+\delta) \norm{\bx}^2$ for all $\bx \in \Sigma_{2k}$, (i.e., the upper bound of the {\em restricted isometry property} (RIP)). In this case the authors show that\footnote{Note that it is possible to remove the assumption that $\bA$ satisfies the upper bound of the RIP, but with a rather unsatisfying result.  Specifically, for an arbitrary matrix $\bA$ with a fixed Frobenius norm, we have that $\norm[2]{\bA}^2 \le \norm[F]{\bA}^2$, so that $(1+\delta) \le \norm[F]{\bA}^2$.  This bound can be shown to be tight by considering a matrix $\bA$ with only one nonzero column.  However, applying this bound underestimates $M^*(\bA)$ by a factor of $n$.  Of course, the bounds coincide for ``good'' matrices (such as random matrices) which will have a significantly smaller value of $\delta$~\cite{RaskuWY_Lower}.  However, the random matrix framework is precisely that which we wish to challenge.}
$$
M^*(\bA) \ge C \frac{k \sigma^2}{(1+\delta) n} \log \left( n/k \right).
$$

Our primary aim, however, is to challenge the use of the RIP and/or random matrices and to determine whether we can do better via a different choice in $\bA$. Our approach relies on standard tools from information theory such as Fano's inequality, and as such is very similar in spirit to the approaches in~\cite{AeroSZ_Information,SarvoBB_Measurements,RaskuWY_Minimax}.  The proof of Theorem~\ref{thm1} begins by following a similar path to that taken in~\cite{RaskuWY_Minimax}.  As in the results of~\cite{RaskuWY_Minimax}, we rely on the construction of a packing set of sparse vectors.  However, we place no assumptions whatsoever on the matrix $\bA$.  To do this we must instead consider a random construction of this set, allowing us to apply the recently established matrix-version of Bernstein's inequality due to Ahlswede and Winter~\cite{AhlswW_Strong} to bound the empirical covariance matrix of the packing set. Our analysis is divided into two parts.  In Section~\ref{sec:proof} we provide the proof of Theorem~\ref{thm1}, and in Section~\ref{sec:packing} we provide the construction of the necessary packing set.

\subsection{Notation}

We now provide a brief summary of the notations used throughout the paper.  If $\bA$ is an $m \times n$ matrix and $T \subset \{1, \ldots, n \}$, then $\bA_T$ denotes the $m \times |T|$ submatrix with columns indexed by $T$.  Similarly, for a vector $\bx \in \real^n$ we let $\bx\restrict{T}$ denote the restriction of $\bx$ to $T$.  We will use $\norm[p]{\bx}$ to denote the standard $\ell_p$ norm of a vector, and for a matrix $\bA$, we will use $\norm[]{\bA}$ and $\norm[F]{\bA}$ to denote the operator and Frobenius norms respectively.

\section{Proof of Main Result}
\label{sec:proof}

In this section we establish the lower bound (\ref{eq:mmriskbound_a}) in Theorem~\ref{thm1}.  The proof of (\ref{eq:mmriskbound_b}) is provided in the Appendix.  In the proofs of both (\ref{eq:mmriskbound_a}) and (\ref{eq:mmriskbound_b}), we will assume that $\sigma = 1$ since the proof for arbitrary $\sigma$ follows by a simple rescaling.  To obtain the bound in (\ref{eq:mmriskbound_a}) we begin by following a similar course as in~\cite{RaskuWY_Minimax}.  Specifically, we will suppose that $\bx$ is distributed uniformly on a finite set of points $\Pset \subset \Sigma_k$, where $\Pset$ is constructed so that the elements of $\Pset$ are well separated. This allows us to apply the following lemma which follows from Fano's inequality combined with the convexity of the Kullback-Leibler (KL) divergence.  We provide a proof of the lemma in the Appendix.

\begin{lemma} \label{lem:fano}
Consider the measurement model where $\by = \bA \bx + \bz$ with $\bz \sim \mathcal{N}(\bzero,\bI)$.  Suppose that there exists set of points $\Pset = \{\bx_i\}_{i=1}^{\absval{\Pset}} \subset \Sigma_k$ such that for any $\bx_i, \bx_j \in \Pset$, $\norm{\bx_i-\bx_j}^2 \ge 8 n M^*(\bA)$, where $M^*(\bA)$ is defined as in (\ref{eq:minimaxrisk}).  Then
\begin{equation} \label{eq:fanoentbound}
\frac{1}{2} \log \absval{\Pset} - 1 \le \frac{1}{2 \absval{\Pset}^2} \sum_{i,j=1}^{\absval{\Pset}} \norm{\bA \bx_i - \bA \bx_j}^2.
\end{equation}
\end{lemma}

\noindent By taking the set $\Pset$ in Lemma~\ref{lem:Sconstruction} below and rescaling these points by $4 \sqrt{n M^*(\bA)}$, we have that there exists a set $\Pset$ satisfying the assumptions of Lemma~\ref{lem:fano} with
$$
\absval{\Pset} = \left(n/k \right)^{k/4},
$$
and hence from (\ref{eq:fanoentbound}) we obtain
\begin{equation} \label{eq:mainbound1}
\frac{k}{4} \log \left( n/k \right) - 2 \le \frac{1}{\absval{\Pset}^2} \sum_{i,j=1}^{\absval{\Pset}} \norm{\bA \bx_i - \bA \bx_j}^2 = \tr{ \bA^*\bA \left(\frac{1}{\absval{\Pset}^2} \sum_{i,j=1}^{\absval{\Pset}} \left(\bx_i -\bx_j\right) \left(\bx_i -\bx_j\right)^* \right)}.
\end{equation}
If we set
$$
\bmu = \frac{1}{\absval{\Pset}} \sum_{i=1}^{\absval{\Pset}} \bx_i  \quad \quad \mathrm{and} \quad \quad \bQ = \frac{1}{\absval{\Pset}} \sum_{i=1}^{\absval{\Pset}}\bx_i \bx_i^*,
$$
then one can show that
$$
\frac{1}{\absval{\Pset}^2} \sum_{i,j=1}^{\absval{\Pset}} \left(\bx_i -\bx_j\right) \left(\bx_i -\bx_j\right)^* = 2 \left( \bQ - \bmu \bmu^* \right).
$$
Thus, we can bound (\ref{eq:mainbound1}) by
$$
2 \, \tr{ \bA^* \bA \left(\bQ - \bmu \bmu^* \right)} \le 2 \, \tr{\bA^* \bA \bQ },
$$
where the inequality follows since $\tr{\bA^* \bA \bmu \bmu^*} = \norm{\bA \bmu}^2 \ge 0$.  Moreover, since $\bA^* \bA$ and $\bQ$ are positive semidefinite,
$$
\tr{\bA^* \bA \bQ } \le \tr{\bA^* \bA} \norm[]{\bQ} = \norm[F]{\bA}^2 \norm[]{\bQ}.
$$
Combining this with (\ref{eq:mainbound1}) and applying Lemma~\ref{lem:Sconstruction} to bound the norm of $\bQ$ --- recalling that it has been appropriately rescaled --- we obtain
$$
\frac{k}{4} \log \left( n/k \right) - 2 \le (1+\beta) 32  M^*(\bA) \norm[F]{\bA}^2,
$$
where $\beta$ is a constant that can be arbitrarily close to $0$. This
yields the desired result.

\section{Packing Set Construction}
\label{sec:packing}

We now return to the problem of constructing the packing set $\Pset$.  As noted above, our construction exploits the following matrix Bernstein inequality of Ahlswede and Winter~\cite{AhlswW_Strong}.  See also~\cite{Tropp_User}.

\begin{thm}[Matrix Bernstein Inequality] \label{thm:matrixBernstein}
Let $\{\bX_i\}$ be a finite sequence of independent zero-mean random self-adjoint matrices of dimension $n \times n$.  Suppose that $\norm[]{\bX_i} \le 1$ almost surely for all $i$ and set $\rho^2 = \sum_{i}\norm[]{\expval{\bX_i^2}}$. Then for all $t \in [0, 2\rho^2]$,
\begin{equation} \label{eq:matrixBernstein}
\prob{\norm[]{\sum_{i} \bX_i} \ge t} \le 2n \exp \left( -\frac{ t^2}{4 \rho^2}  \right).
\end{equation}
\end{thm}

We construct the set $\Pset$ by choosing points at random, which allows us to apply Theorem~\ref{thm:matrixBernstein} to establish a bound on the empirical covariance matrix. In bounding the size of $\Pset$ we follow a similar course as in~\cite{RaskuWY_Minimax} and rely on techniques from~\cite{Kuhn_Lower}.

\begin{lemma} \label{lem:Sconstruction}
Let $n$ and $k$ be given, and suppose for simplicity that $k$ is even and $k < n/2$.  There exists a set $\Pset = \{x_i\}_{i=1}^{\absval{\Pset}} \subset \Sigma_k$ of size
\begin{equation} \label{eq:Scard}
\absval{\Pset} = \left( n/k \right)^{k/4}
\end{equation}
such that
\begin{enumerate}[(i)]
\item $\norm{\bx_i - \bx_j}^2 \ge 1/2$ for all $\bx_i, \bx_j \in \Pset$ with $i \neq j$; and
\item $\norm[]{\frac{1}{\absval{\Pset}} \sum_{i=1}^{\absval{\Pset}} \bx_i \bx_i^* - \frac{1}{n} \bI} \le \beta / n$,
\end{enumerate}
where $\beta$ can be made arbitrarily close to $0$ as $n \rightarrow \infty$.
\end{lemma}
\begin{proof}
We will show that such a set $\Pset$ exists via the probabilistic method.  Specifically, we will show that if we draw $\absval{\Pset}$ independent $k$-sparse vectors at random, then the set will satisfy both {\em (i)} and {\em (ii)} with probability strictly greater than 0.  We will begin by considering the set
$$
\mathcal{U} = \left\{ \bx \in \left\{0,+ \sqrt{1/k}, - \sqrt{1/k} \right\}^n : \norm[0]{\bx} = k \right\}.
$$
Clearly, $\absval{\mathcal{U}} = \binom{n}{k}2^k$.  Next, note that for all $\bx,\bx' \in \mathcal{U}$, $\frac{1}{k} \norm[0]{\bx' - \bx} \le \norm{\bx' - \bx}^2$, and thus if $\norm{\bx' - \bx}^2 \le 1/2$ then $\norm[0]{\bx' - \bx} \le k/2$.  From this we observe that for any fixed $\bx \in \mathcal{U}$,
$$
\left| \left\{ \bx' \in \mathcal{U} :  \norm{\bx'-\bx}^2 \le 1/2 \right\} \right| \le \left| \left\{ \bx' \in \mathcal{U} : \norm[0]{\bx'-\bx} \le k/2 \right\} \right| \le \binom{n}{k/2} 3^{k/2}.
$$
Suppose that we construct $\Pset$ by picking elements of $\mathcal{U}$ uniformly at random.  When adding the $j^{\mathrm{th}}$ point to $\Pset$, the probability that $\bx_j$ violates {\em (i)} with respect to the previously added points is bounded by
$$
\frac{(j-1) \binom{n}{k/2} 3^{k/2} }{\binom{n}{k} 2^k }.
$$
Thus, using the union bound, we can bound the total probability that $\Pset$ will fail to satisfy {\em (i)}, denoted $P_1$, by
$$
P_1 \le \sum_{j=1}^{\absval{\Pset}} \frac{(j-1) \binom{n}{k/2} 3^{k/2} }{\binom{n}{k} 2^k } \le \frac{\absval{\Pset}^2}{2} \frac{\binom{n}{k/2}}{\binom{n}{k}} \left(\frac{\sqrt{3}}{2} \right)^k.
$$
Next, observe that
$$
\frac{\binom{n}{k}}{\binom{n}{k/2}} = \frac{(k/2)! (n -k/2)!}{k! (n-k)!} = \prod_{i=1}^{k/2} \frac{n - k + i}{k/2 + i} \ge \left( \frac{n- k + k/2}{k/2 + k/2} \right)^{k/2} = \left( \frac{n}{k} - \frac{1}{2} \right)^{k/2},
$$
where the inequality follows since $(n-k+i)/(k/2 + i)$ is decreasing as a function of $i$ provided that $n-k > k/2$. Also,
$$
\left( \frac{n}{k} \right)^{k/2} \left( \frac{\sqrt{3}}{2} \right)^k = \left( \frac{3n}{4k} \right)^{k/2} \le \left(\frac{n}{k} - \frac{1}{2} \right)^{k/2}
$$
with the proviso $k \le n/2$.  Thus, for $\absval{\Pset}$ of size
given in (\ref{eq:Scard}),
\begin{equation} \label{eq:P1}
P_1 \le \frac{1}{2} \left( \frac{n}{k} \right)^{k/2} \frac{\binom{n}{k/2}}{\binom{n}{k}} \left(\frac{\sqrt{3}}{2} \right)^k \le \frac{1}{2} \left(\frac{n}{k} - \frac{1}{2} \right)^{k/2} \frac{\binom{n}{k/2}}{\binom{n}{k}} \le \frac{1}{2} \frac{\binom{n}{k}}{\binom{n}{k/2}} \frac{\binom{n}{k/2}}{\binom{n}{k}} \le \frac{1}{2}.
\end{equation}

Next, we consider {\em (ii)}.  We begin by letting
$$
\bX_i = \bx_i \bx_i^* - \frac{\bI}{n}.
$$
Since $\bx_i$ is drawn uniformly at random from $\mathcal{U}$, it is straightforward to show that $\norm[]{\bX_i} \le 1$ and that $\expval{\bx_i \bx_i^*} = \bI / n$, which implies that $\expval{\bX_i} = 0$. Moreover,
$$
\expval{\bX_i^2} = \expval{ \left( \bx_i \bx_i^* \right)^2 } - \left( \frac{1}{n} \bI \right)^2 = \frac{(n-1)}{n^2}\bI.
$$
Thus we obtain $\rho^2 = \sum_{i=1}^{\absval{\Pset}} \norm[]{\expval{\bX_i^2}} = \absval{\Pset} (n-1)/n^2 \le \absval{\Pset} / n$. Hence, we can apply Theorem~\ref{thm:matrixBernstein} to obtain
$$
\prob{ \norm[]{ \sum_{i=1}^{\absval{\Pset}} \bX_i  }  \ge t } \le 2 n \exp \left( - \frac{t^2 n}{4 \absval{\Pset}} \right).
$$
Setting $t = \absval{\Pset} \beta /n$, this reduces to show that the probability that $\Pset$ will fail to satisfy {\em (ii)}, denoted $P_2$, is bounded by
$$
P_2 \le 2n \exp \left( - \frac{\beta^2 \absval{\Pset} }{4 n} \right).
$$

For the lemma to hold we require that $P_1+P_2 < 1$, and since $P_1 < \frac{1}{2}$ it is sufficient to show that $P_2 < \frac{1}{2}$.  This will occur provided that
$$
\beta^2 > \frac{4n \log(4n) }{\absval{\Pset}}.
$$
Since $\absval{\Pset} = \Theta\left( (n/k)^k \right)$, $\beta$ can be made arbitrarily small as $n \rightarrow \infty$.
\end{proof}

\section*{Appendix}

\begin{proof}[Proof of (\ref{eq:mmriskbound_b}) in Theorem~\ref{thm1}]

We begin by noting that
$$
M^*(\bA) = \inf_{\xhat} \; \sup_{T : |T| \le k} \; \sup_{\bx : \supp(\bx) = T} \expval{\frac{1}{n} \norm{\xhat(\by) - \bx}^2} \ge \sup_{T : |T| \le k} \; \inf_{\xhat} \; \sup_{\bx : \supp(\bx) = T} \expval{\frac{1}{n} \norm{\xhat(\by) - \bx}^2}.
$$
Thus for the moment we restrict our attention to the subproblem of bounding
\begin{equation} \label{eq:subprob}
M^*(\bA_T) = \inf_{\xhat} \; \sup_{\bx : \supp(\bx) = T} \expval{\frac{1}{n}\norm{\xhat(\by) - \bx}^2} = \inf_{\xhat} \; \sup_{\bx \in \real^k} \expval{\frac{1}{n} \norm{\xhat(\bA_T \bx + \bz) - \bx}^2},
\end{equation}
where $\xhat(\cdot)$ takes values in $\real^k$.  The last equality of (\ref{eq:subprob}) follows since if $\supp(\bx) = T$ then
$$
\norm{\xhat(\by) - \bx}^2 = \norm{\xhat(\by)\restrict{T} - \bx\restrict{T}}^2 + \norm{\xhat(\by)\restrict{T^c}}^2,
$$
so that the risk can always be decreased by setting $\xhat(\by)|_{T^c} = 0$.  This subproblem (\ref{eq:subprob}) has a well-known solution (see Exercise 5.8 on pp.\ 403 of~\cite{LehmaC_Theory}).  Specifically, let $\lambda_i(\bA_T^* \bA_T)$ denote the eigenvalues of the matrix $\bA_T^* \bA_T$.  Then
\begin{equation} \label{eq:subprob3}
M^*(\bA_T) = \frac{1}{n} \sum_{i=1}^k \frac{1}{\lambda_i(\bA_T^* \bA_T)}.
\end{equation}
Thus we obtain
\begin{equation} \label{eq:bound1}
M^*(\bA) \ge \sup_{T : |T| \le k} M^*(\bA_T) = \sup_{T : |T| \le k} \frac{1}{n} \sum_{i=1}^k \frac{1}{\lambda_i(\bA_T^* \bA_T)}.
\end{equation}
Note that if there exists a subset $T$ for which $\bA_T$ is not full rank, then at least one of the eigenvalues $\lambda_i(\bA_T^* \bA_T)$ will vanish and the minimax risk will be unbounded.  This also shows that the minimax risk is always unbounded when $m < k$.

Thus, we now assume that $\bA_T$ is full rank for any choice of $T$. Since $f(x) = 1/x$ is a convex function for $x > 0$, we have that
\begin{equation*}
\sum_{i=1}^k \frac{1}{\lambda_i(\bA_T^* \bA_T)} \ge \frac{k^2}{\sum_{i=1}^k \lambda_i(\bA_T^* \bA_T) } = \frac{k^2}{\norm[F]{\bA_T}^2}.
\end{equation*}
Since there always exists a set of $k$ columns $T_0$ such that $\norm[F]{\bA_{T_0}}^2 \le (k/n)\norm[F]{\bA}^2$,  (\ref{eq:bound1}) reduces to yield the desired result.

\end{proof}

\begin{proof}[Proof of Lemma~\ref{lem:fano}]
To begin, note that if $\bx$ is uniformly distributed on the set of points in $\Pset$, then there exists an estimator $\xhat(\by)$ such that
\begin{equation} \label{eq:probbound}
\expval[\bx,\bz]{\frac{1}{n} \norm{\xhat(\by) - \bx}^2} \le M^*(\bA),
\end{equation}
where the expectation is now taken with respect to both the signal and
the noise.  We next consider the problem of deciding which $\bx_i \in
\Pset$ generated the observations $\by$.  Towards this end, set
$$
T(\xhat(\by)) = \mathop{\arg \min}_{\bx_i \in \Pset} \norm{\xhat(\by) - \bx_i}.
$$
Define $P_e = \prob{T(\xhat(\by)) \neq \bx}$.  From Fano's inequality~\cite{CoverT_Elements} we have that
\begin{equation} \label{eq:fano}
H(\bx|\by) \le 1 + P_e \log \absval{\Pset}.
\end{equation}

We now aim to bound $P_e$.  We begin by noting that for any $\bx_i \in \Pset$ and any $\xhat(\by)$, $T(\xhat(\by)) \neq \bx_i$ if and only if there exists an $\bx_j \in \Pset$ with $j \neq i$ such that
$$
\norm{\xhat(\by) - \bx_i} \ge \norm{\xhat(\by) - \bx_j} \ge \norm{\bx_i - \bx_j} - \norm{\xhat(\by) - \bx_i}.
$$
This would imply that
$$
2 \norm{\xhat(\by) - \bx_i} \ge \norm{\bx_i - \bx_j} \ge \sqrt{ 8n M^*(\bA)}.
$$
Thus, we can bound $P_e$ using Markov's inequality as follows:
$$
P_e \le \prob{ \norm{\xhat(\by) - \bx_i}^2 \ge 8 n M^*(\bA)/4 } \le \frac{\expval[\bx,\bz]{\norm{\xhat(\by) - \bx_i}^2}}{2 n M^*(\bA)} \le \frac{nM^*(\bA)}{2 n M^*(\bA)} = \frac{1}{2}.
$$
Combining this with (\ref{eq:fano}) and the fact that $H(\bx) = \log \absval{\Pset}$, we obtain
$$
I(\bx,\by) = H(\bx) - H(\bx|\by) \ge \frac{1}{2} \log \absval{\Pset} - 1.
$$
From the convexity of KL divergence (see~\cite{HanV_Generalizing} for details), we have that
$$
I(\bx,\by) \le \frac{1}{\absval{\Pset}^2} \sum_{j,k=1}^{\absval{\Pset}} D \left(\cP_i , \cP_j \right),
$$
where $D \left(\cP_i , \cP_j \right)$ represents the KL divergence from $\cP_i$ to $\cP_j$ where $\cP_i$ denotes the distribution of $\by$ conditioned on $\bx = \bx_i$.  Since $\bz \sim \cN(\bzero,\bI)$, $\cP_i$ is simply given by $\cN( \bA \bx_i, \bI)$.  Standard calculations demonstrate that $D \left(\cP_i , \cP_j \right) = \frac{1}{2} \norm{\bA \bx_i - \bA \bx_j}^2$, establishing (\ref{eq:fanoentbound}).
\end{proof}

\bibliographystyle{plain}
\footnotesize
\bibliography{bibpreamble,bibmain}

\end{document}